\documentclass[a4paper,11pt,twoside]{article}
\usepackage[a4paper,left=2.54cm,right=2.54cm,top=2.54cm,bottom=2.54cm]{geometry}
\usepackage{amsthm,amsfonts,amsmath}
\usepackage{harvard}
\usepackage{graphicx}
\allowdisplaybreaks

\newcommand{\bbr}{\mathbb{R}}  
\newcommand{\bbn}{\mathbb{N}}

\newcommand{\fn}{\footnote}
\newcommand{\ci}{\citeasnoun}

\newcommand{\essinf}{{\rm \, essinf~}}

\newcommand{\fil}{\mathcal{F}}
\newcommand{\fcal}{\fil}

\newcommand{\xcal}{\mathcal{X}}

\newcommand{\acal}{\mathcal{A}}

\newcommand{\be}{\begin{equation}}
\newcommand{\ee}{\end{equation}}
\newcommand{\bew}{\begin{equation*}}
\newcommand{\eew}{\end{equation*}}
\newcommand{\1}{{\bf 1}}

\newtheorem{thm}{Theorem}
\newtheorem{defi}[thm]{Definition}

\newtheorem{corollary}[thm]{Corollary}

\newtheorem{rem}[thm]{Remark}
\newtheorem{ex}[thm]{Example}

\begin{document}


\title{Solvency II, or How to Sweep the Downside Risk Under the Carpet}

\author{Stefan Weber\footnote{Institut f\"ur Mathematische Stochastik, Leibniz Universit\"at Hannover, Welfengarten 1, 30167 Hannover, Germany.
e-mail:  {\tt sweber@stochastik.uni-hannover.de}.  \newline The author would like to thank Paul Embrechts, Anna-Maria Hamm, Thomas Knispel, Pablo Koch Medina, Andreas M\"arkert, Michael Schmutz, Ruodu Wang and an anonymous referee for their helpful comments.} \\[1.0ex] \textit{Leibniz Universit{\"a}t Hannover}}

\date{\today
 }

\maketitle

\begin{abstract}
Under Solvency II the computation of capital requirements is based on value at risk (V@R). V@R is a quantile-based risk measure and neglects extreme risks in the tail. V@R belongs to the family of distortion risk measures. A serious deficiency of V@R is that firms can hide their total downside risk in corporate networks, unless a consolidated solvency balance sheet is required  for each economic scenario. In this case, they can largely reduce their total capital requirements via appropriate transfer agreements within a network structure consisting of sufficiently many entities and thereby circumvent capital regulation. We prove several versions of such a result for general distortion risk measures of V@R-type, explicitly construct suitable allocations of the network portfolio, and finally demonstrate how these findings can be extended beyond distortion risk measures. We also discuss why consolidation requirements cannot completely eliminate this problem. Capital regulation should thus be based on coherent or convex risk measures like average value at risk or expectiles. 

\end{abstract}\vspace{0.2cm}

\section{Introduction}\label{sec:intro}

Capital requirements are a key instrument in the regulation of financial institutions. Their computation is typically based on two ingredients: stochastic balance sheet projections as a description of a firm's business, and monetary risk measures that capture the normative standards of a regulator. The question which risk measure to use for regulation is the topic of an ongoing discussion between academics and practitioners that began in the mid 1990s. Different properties of monetary risk measures have been suggested, and corresponding classes of risk measures have been identified and characterized.

Most of the scientific literature deals with convex risk measures: convex risk measures assign a lower risk measurement to a diversified position than to the non-diversified positions from which the diversified position is composed. If a convex risk measure is also positively homogeneous, it is subadditive; this property facilitates the delegation of risk limits from the company level to individual departments of the firm. Moreover, convex risk constraints are technically easier to handle in the context of portfolio optimization than non-convex constraints.

As this paper will show, risk measures that are not convex may have additional deficiencies. The European regulatory framework for insurance firms, Solvency II, is built on the non-convex risk measure value at risk (V@R). In the current paper, we define a broader class of risk measures, called V@R-type risk measures, that includes V@R as a special case, and prove that a sophisticated firm can completely circumvent capital regulation that is based on these risk measures, if capital requirements are not also computed on the basis of a consolidated solvency balance sheet. We discuss why such consolidation requirements might not easily be enforced for all potentially relevant entities.  

The main ideas is the following. A firm can adjust its structure over time and form a collection of multiple legally separate entities. We assume that these entities are designed in such a way that they are regulated individually.\footnote{Our results can be adapted to a situation in which only a fraction of the entities is regulated individually.} This might be achieved by designing a sophisticated network of firms that includes suitable entities outside the European Economic Area (EEA) and respects legislation on consolidated accounts, e.g. Directive 83/349/EEC of the \ci{EC_Cons}. The shareholders seek to minimize the total capital required for the operation of the whole business.\footnote{\ci{embrechtsliuwang2016} show in the special case of Range Value at Risk that the solution to this objective equals an equilibrium in an appropriate market. Central control is in this case not required to establish an appropriate solution. We conjecture that this result can be generalized to larger classes of risk measures.}

Assuming that the future net asset value of the company is described by a random variable $E$, a sophisticated firm can create a corporate network of $n$ subentities and split $E$ into $n$ parts. For this purpose, it needs to design suitable legally binding transfer agreements that produce an allocation $(E^1, E^2, \dots, E^n)$ of the total net asset value among the subentities, satisfying $\sum_{i=1}^n E^i = E$. We show that for V@R-type risk measures and sufficiently many subentities the total capital requirement can thereby be reduced to the capital requirement of a corporate network with a deterministic future net asset value of ${\rm esssup~}E$, corresponding to the best case scenario. If the risk measure of one of the subentities is in addition strongly surplus sensitive -- a property that we define in this paper -- and if at the same time unlimited leverage is admissible, then the total capital requirement of the network can be reduced to levels that converge to minus infinity as the leverage of suitable subentities approaches infinity. 

The paper is structured as follows: Section \ref{sec:capital} reviews capital regulation and then recalls the family of distortion risk measures that includes value at risk (V@R), average value at risk (AV@R) and range value at risk (RV@R) as special cases. In Section \ref{sec:group} we first explain in detail the relationship between solvency capital minimization and optimal risk sharing. Second, we describe in the context of distortion risk measures how the total capital requirement of a corporate network can be reduced in order to circumvent capital regulation. Appropriate allocations for the corporate network are  constructed. In the case of distortion risk measures, we provide explicit formulas in terms of mixtures of V@R for the total capital requirement of these allocations. Third, we generalize the main results beyond the case of distortion risk measures. All proofs are collected in the Appendix.

\subsection*{Literature}

Risk sharing for V@R is considered in \ci{Gal10}, for convex distortion risk measures in \ci{Jouini2008}. This paper is most closely related to \ci{embrechtsliuwang2016} who investigate the risk sharing problem for a two-parameter class of quantile-based risk measures, called range value at risk (RV@R). This family, introduced by \ci{Rama}, includes V@R and AV@R as limiting cases. Our paper, in contrast, provides a general picture on risk sharing for V@R-type risk measures -- a notion that is introduced in the current paper -- and includes the main results of \ci{embrechtsliuwang2016} as special cases. A preliminary extension of the results of \ci{embrechtsliuwang2016} can also be found in \ci{Agirman}. Risk sharing of homogenous agents that evaluate risk with a common risk measure is studied in \ci{Wang2016}. The paper considers the self-convolution and subadditive hull of a risk measure and categorizes risk measures according to the size of the regulatory arbitrage that they admit. For further related references on optimal risk sharing we refer to \ci{embrechtsliuwang2016}. Our paper confirms the economic implications that are discussed in \ci{embrechtsliuwang2016}  and \ci{Wang2016}. A description and analysis of corporate groups can be found in \ci{Keller2007}, \ci{Filipovic} and \ci{Haier}.

Some of our key arguments rely on the representation of distortion risk measures as mixtures of V@R, as described in \ci{Dhaene2012}. For reviews on the theory of monetary risk measures, including distortion risk measures, we refer to \ci{foellmer-schied3rd} and \ci{FW15}. For further results on Choquet integrals and distortion risk measures see \ci{Choquet}, \ci{Dhaene06}, \ci{Denneberg}, \ci{Greco}, \ci{Schmeidler}, \ci{Song2006}, \ci{Song2009},  \ci{Song2009459}, and \ci{wang1996}.  The current paper uses the same sign convention for risk measures as \ci{embrechtsliuwang2016}.

\section{Capital Regulation}\label{sec:capital}

\subsection{Solvency II}

A key instrument in the regulation of financial firms such as insurance companies  and banks are solvency capital requirements. Their main role is to provide a buffer for potential losses that protects customers, policy holders and other counterparties. Solvency II is the regulatory framework that applies to European insurance companies. The computation of capital requirements is described in the  Directive 2009/138/EC of the European Parliament and of the Council on
the taking-up and pursuit of the business of Insurance and Reinsurance -- Solvency II (see \ci{EC}):
\begin{quote}
The Sol­vency Capital Requirement should be determined as the economic capital to be held by insurance and reinsurance undertakings in order to ensure that ruin occurs no more often than once in every 200 cases or, alternatively, that those undertakings will still be in a position, with a prob­ability of at least 99.5 \%, to meet their obligations to policy holders and beneficiaries over the following 12 months. That economic capital should be calculated on the basis of the true risk profile of those undertakings, taking account of the impact of possible risk-mitigation techniques, as well as diversification effects.
\end{quote} 

In a stylized manner, these  principles can be formalized as follows: Consider an atomless probability space $(\Omega, \fil, P)$ and a one period economy with dates $t=0,1$. Time 0 will be interpreted as today, time 1 as the one-year time horizon of Solvency II. Suppose that the solvency balance sheet of an insurance firm is available for $t=0,1$, e.g. computed from available data using an internal model. The value of the assets at time $t=0,1$ is denoted by $A_t$. We set $L_t$, $t=0,1$, for the value of the total liabilities to customers and other counterparties, net of the book value of equity. The book value of equity or net asset value (NAV) is then computed as the difference of assets and liabilities, i.e. $E_t = A_t - L_t$, $t=0,1$. Observe that in this situation quantities at time $0$ are deterministic while quantities at time $1$ are random. For simplicity, we neglect the risk-less interest rates over this time horizon.\footnote{For a discussion of this issue, we refer to \ci{christiansen2014}.}

Directive 2009/138/EC states that capital must be sufficient to prevent ruin with probability $99.5\%$ on a one-year time horizon, i.e. $P (E_1 < 0) \leq \alpha$ with $\alpha = 0.5\%$. Setting 
\begin{equation}\label{eq:SCR_V@R}
SCR: = V@R_\alpha (- \Delta E_1)
\end{equation}
for $\Delta E_1 = E_1 - E_0$, we find conditions that are equivalent to the solvency requirement:
\begin{equation}\label{eq:eqSCR_V@R}
P(E_1 < 0)   \leq   \alpha \quad  \Leftrightarrow \quad - E_1 \in \acal _ {V@R_\alpha} \quad \Leftrightarrow \quad SCR \leq E_0 
\end{equation}
where $\acal_{V@R_\alpha} = \{X\in L^\infty: P(X > 0) \leq \alpha\}$ is the acceptance set of $V@R_\alpha$ and $V@R_\alpha(X) = \inf \{m\in \bbr: X-m \in \acal_{V@R_\alpha}\}$. Observe that -- in contrast to \ci{foellmer-schied3rd}, but consistent with \ci{embrechtsliuwang2016} --  we make the convention that the argument of V@R counts losses  positive and profits negative. 
We would like to point out that Directive 2009/138/EC provides an acceptance set for the company's capital at time $t=1$. This is equivalent to verifying that the SCR is less than firm's capital at time $t=0$ where the SCR is computed by the risk measure that corresponds to this acceptance set, evaluated at the random capital increment $- \Delta E_1 $ (due to our sign convention). 

While Solvency II limits the ruin probability at the one-year time horizon  -- corresponding to the acceptance set of V@R -- this specific criterion can easily be replaced by others, i.e. by possibly more desirable acceptance sets. Then the modified SCR must be computed as the corresponding risk measure evaluated at the random capital increment $- \Delta E_1 $. Examples include the Swiss Solvency test and Basel III. Both are based on AV@R, also called expected shortfall, conditional value at risk, tail value at risk, or tail conditional expectation. The next section recalls distortion risk measures that include both V@R and AV@R as special cases. 

\subsection{Distortion Risk Measures}

To begin with, let $\xcal$ be the space of measurable and bounded functions on a measurable space $(\Omega, \fil)$.\footnote{The essential domain of risk measures, defined on larger spaces, relies on each risk measure itself. To keep the presentation simple, we first limit our attention to bounded measurable functions, but explain later -- in Remark~\ref{domains} -- how larger domains may be chosen.}
A risk measure $\rho: \xcal \to \bbr$ is a monotone and cash-invariant function, see \ci{foellmer-schied3rd}, Definition 4.1: 

\begin{enumerate} 
\item  \emph{Monotonicity}:  $\quad \quad X,Y\in \xcal, X\leq Y \; \Rightarrow  \; \rho (X) \leq \rho (Y)$
\item \emph{Cash-Invariance}: $\quad X\in \xcal, m\in \bbr \; \Rightarrow \;  \rho(X+m) = \rho (X) + m $
\end{enumerate}
We use the convention that losses are counted positive and gains negative. A risk measure $\rho$ is normalized, if $\rho(0) = 0$. It is distribution-based, if $\xcal$ is a space of random variables on some probability space $(\Omega, \fil,P)$ and $\rho(X)=\rho(Y)$ whenever the distributions of $X$ and $Y$ under $P$ are equal, i.e. $P^X=P^Y$, for $X, Y \in \xcal$. 

Any risk measure  corresponds to its acceptance set, $\acal = \{ X \in \xcal: \rho(X) \leq 0 \}$, from which it can be recovered as a capital requirement:
$$\rho (X) = \inf \{m \in \bbr: X - m \in \acal\} . $$
Using the notation of the previous section, solvency capital requirements are described as follows: If a regulator requires $- E_1 \in \acal$, this is equivalent to $SCR \leq E_0$ with $SCR := \rho (- \Delta E _1)$.

We will now focus on a specific family of risk measures: distortion risk measures. Some results can be extended beyond this setting, see Section \ref{sec:V@Rtype}.  Distortion risk measures form a subset of the family of comonotonic risk measures. The latter can be expressed as Choquet integrals with respect to capacities. Here we recall the main results that we will need. 

\begin{defi}
Two measurable functions $X, Y$  on $(\Omega, \fil)$ are \emph{comonotonic} if
$$(X(\omega) - X(\omega '))(Y(\omega) - Y(\omega ')) \geq 0 \quad \forall (\omega, \omega') \in \Omega \times \Omega .$$
A risk measure $\rho: \xcal \to \bbr$ is \emph{comonotonic} if 
$$\rho(X+ Y) = \rho (X) + \rho (Y)$$
for comonotonic $X, Y \in \xcal$.
\end{defi}

\begin{rem}
\begin{enumerate}
\item All comonotonic risk measures are positively homogeneous. 
\item V@R and AV@R are comonotonic. \ci{Rama} suggest an alternative to V@R and AV@R, called range value at risk (RV@R), which is a further example of a comonotonic risk measure. Letting $\alpha, \beta >0$ with $\alpha + \beta \leq 1$, they define  
$$RV@R_{\alpha, \beta} (X) = \frac {1}{\beta} \int _{\alpha}^{\alpha + \beta} V@R_\lambda (X) d \lambda $$
for $X\in\xcal$. Like V@R, this is a non convex risk measure with an index of qualitative robustness\footnote{For a precise definition see \ci{Kratschmer2014}.} (IQR) of $\infty$, while AV@R is convex with IQR of 1. Observe that for convex risk measures, the IQR is at most 1. The limiting cases of $RV@R_{\alpha, \beta}$ correspond to $V@R_{\alpha}$ for $\beta \to 0$ and $AV@R_\beta$ for $\alpha \to 0$. 
\end{enumerate}
\end{rem}

\begin{defi}\label{def:choquet}
\begin{enumerate}
\item A mapping $c: \fcal \rightarrow [0,\infty)$ is called a \emph{monotone set function} if it satisfies the following properties:
\begin{enumerate}
\item $c(\emptyset) = 0$.
\item $A, B \in \fcal$, $A\subseteq B$ $\Rightarrow$ $c(A) \leq c(B)$. 
\end{enumerate}
If, in addition, $c(\Omega) =1$, i.e. $c$ is normalized, then $c$ is called a \emph{capacity}.
\item Let $X \in \xcal$. The \emph{Choquet integral} of $X$ with respect to the monotone set function $c$ is defined by
$$\int X dc \; = \;  \int _{-\infty}^0 [c(X>x) - c(\Omega) ] dx  \; + \;  \int _0^\infty c(X>x) dx$$
\end{enumerate}
\end{defi}
The Choquet integral coincides with the usual integral if $c$ is a $\sigma$-additive probability measure. The following characterization theorem can, for example, be found in Chapter 4 of \ci{foellmer-schied3rd}. 
\begin{thm}
A monetary risk measure $\rho: \xcal \to \bbr$ is comonotonic, if and only if there exists a capacity $c$ on $(\Omega, \fil)$ such that 
$$\rho (X) = \int X dc . $$
\end{thm}

\begin{rem}\label{distortion}
An important special case are distortion risk measures. In this case, the capacity is defined in terms of a distorted probability measure $P$. The resulting capacity is absolutely continuous with respect to $P$, but typically not additive. 
\begin{enumerate}
\item An increasing function $g: [0,1] \to [0,1]$ with $g(0)=0$ and $g(1) =1$ is called a \emph{distortion function}. If $P$ is a probability measure on $(\Omega, \fil)$, then 
$$c^g (A) : = g(P[A]), \quad A \in \fil ,   $$
defines a capacity.
\item
The corresponding \emph{distortion risk measure} $\rho^g (X) := \int  X dc^g $ is coherent, if and only if $g$ is concave.
\item If an increasing function $g: [0,1] \to [0,\infty)$ with $g(0) = 0$ does not satisfy $g(1)=1$, the equation $c^g  (A) = g(P[A]),$ $A \in \fil$, still defines a monotone set function, but $c^g$ is not normalized. 
\end{enumerate}
\end{rem}
\begin{defi}
Consider the class of  distortion functions $g$ such that 
$$
\begin{array}{ll}
g(x) = 0, & \forall x \in [0,\alpha] \\
g(x) >0,  & \forall x \in (\alpha, 1]
\end{array}
$$
for some $\alpha\in [0,1)$. The number $\alpha$ is called the \emph{parameter} of $g$, and 
$$\hat g (x) = 
\left\{ 
\begin{array}{ll}
g(x+ \alpha), & 0 \leq x \leq 1- \alpha\\
1, & 1- \alpha < x
\end{array}
\right.$$
is the \emph{active part} of $g$. 
If the parameter $\alpha > 0$, then $\rho^g$ is called a \emph{V@R-type distortion risk measure}. 
\end{defi}
V@R, AV@R and RV@R are distortion risk measures. V@R and RV@R are of V@R-type, AV@R is not. This is shown in the Table~\ref{tab:distortion}.

\begin{table}[h]
	\begin{center}
 \begin{tabular}{|l|c|c|c|}
 \hline
 \bf  Risk Measure &   $V@R_\alpha$ &  $AV@R_\beta$  &  $RV@R_ {\alpha, \beta}$  \\ \hline \hline
\bf g (x) ~~~= &  $ 
\left\{ 
\begin{array}{ll}
 0 , & 0 \leq x \leq \alpha\\
1, & \alpha < x
\end{array}
\right.$ &
$ 
\left\{ 
\begin{array}{ll}
 \frac{x}{\beta} , & 0 \leq x \leq \beta\\
1, & \beta < x
\end{array}
\right.$
&
$ 
\left\{ 
\begin{array}{ll}
 0 , & 0 \leq x \leq \alpha\\
 \frac {x-\alpha} {\beta}, & \alpha < x \leq \alpha + \beta \\
1, & \alpha + \beta < x
\end{array}
\right.$ \\ \hline
\bf Type &   V@R-type &   {\bf Not} V@R-type & V@R-type \\ \hline
\end{tabular}
\caption{Distortion functions for the risk measures V@R, AV@R and RV@R for $\alpha, \beta >0$ with $\alpha + \beta \leq 1$. }
\label{tab:distortion}
\end{center}
\end{table}

\begin{rem}\label{rem:repre}
Distortion risk measures can be expressed as mixtures of V@R. For arbitrary distortion functions the precise result is described in \ci{Dhaene2012}. In this paper, we will focus only on  the left-continuous case. 
Let $\rho^g$ be defined as in Remark \ref{distortion} for a left-continuous distortion function $g$, then
$$\rho^g(X) \;  = \;  \int_{[0,1]} V@R_\lambda (X) g(d\lambda).$$
The integral on the right hand side of this equation is a Lebesgue-Stieltjes-integral with respect to the function $g$.

This representation provides an interpretation of the parameter $\alpha$ of the distortion function $g$ of a V@R-type distortion risk measure. The distortion risk measure evaluated at $X\in\xcal$ can be written as $\rho^g(X) = \int_{[\alpha,1]} V@R_\lambda (X) g(d\lambda)$ showing that this risk measurement does not depend on any properties of the tail of $X$ beyond its V@R at level $\alpha$. 
\end{rem}

\section{Network Risk Minimization}\label{sec:group}

Financial institutions are typically owned by shareholders with limited liability. The free surplus that can be distributed as dividends to the shareholders is the NAV less the SCR. Shareholders are thus interested in reducing the SCR via appropriate risk management techniques. Generalizing the results of \ci{embrechtsliuwang2016}, we show that corporate network structures with sufficiently many entities admit a  reduction of the total SCR of the network to the SCR of the best case scenario, if capital regulation is based on V@R-type distortion risk measures. This relies on the assumption that the individual SCRs of the entities are added up to obtain the network's SCR; this means in particular that the networks's SCR is not computed on the basis of a consolidated solvency balance sheet. We discuss this assumption in Remark~\ref{consol}. 

We provide an upper bound for the optimal SCR for any number of entities in the corporate network and explicitly construct a network portfolio allocation that attains this bound. If the active parts of the considered distortion functions are concave, we show that the bound is sharp and the corresponding allocation is optimal. We also prove that, if losses and profits are allowed to be unbounded, the total capital requirement of the network may be reduced to any level, provided that one of the risk measures is strongly surplus sensitive. Finally, we demonstrate that our main results are not limited to the family of distortion risk measures. A reduction of the total SCR to the SCR of the best case scenario is in fact possible for all V@R-type risk measures in corporate networks that consist of sufficiently many entities; a reduction to an arbitrarily small level is admissible under conditions that we will specify.

\subsection{The risk sharing problem of the network}

Consider a financial corporation that consists of $n$ entities that are all individually subject to capital regulation. The corporate network is, however, contractually structured in such a way that it serves the same equity holders.\footnote{As discussed before, \ci{embrechtsliuwang2016} give arguments that this assumption might not be necessary.} Over short time horizons the number of entities $n$ is fixed, but the corporation may adjust its structure over longer time horizons. Suppose that the total consolidated assets and liabilities at times $t=0,1$ are given by $A_t$ and $L_t$, respectively. The total NAV is, thus, given by $E_t = A_t -  L_t$. We set $X = - E_1 = L_1 - A_1$. 

The corporate network now uses at time $t=0$ legally binding transfer agreements to modify the NAVs at time $t=1$. In contrast to \ci{Filipovic}, we do not assume that these transfers are constructed as linear portfolios of a finite family of standardized capital transfer products. Instead we suppose that transfer agreements are contingent claims that admit any reallocation of total capital among the $n$ entities of the network. The resulting allocation will be denoted by $(E^i_1)_{i=1,2, \dots, n}$. We set $X^i = - E^i_1$, $i=1,2, \dots, n$, and observe that 
$$X \; = \;  \sum _{i=1} ^n X^i . $$
We suppose that the solvency capital requirement $SCR^i$ of entity $i=1,2,\dots, n$ is computed on the basis of a risk measure $\rho^i$, i.e.
$$SCR^i = E^i_0 + \rho ^i (X^i),$$
where $E_0^i$ refers to the NAV of entity $i$ at time $0$. It holds that $\sum_{i=1}^n E_0^i = E_0$.
The total solvency capital requirement of the network is thus given by 
$$\sum_{i=1}^n SCR^i = E_0 + \sum _{ i=1}^n \rho^i (X^i) .$$
For a fixed number $n$ of entities  the problem of the corporate network consists in the design of optimal transfers that minimize $\sum _{ i=1}^n \rho^i (X^i) $. We will, in particular, show that for V@R-type risk measures and sufficiently large $n$, the corporate network can find a capital allocation such that 
$$ \sum _{ i=1}^n \rho^i (X^i) \; = \; \essinf X  \; = \; -\;  {\rm esssup} \; E_1, $$
corresponding to the best case scenario.  If one of the risk measures is surplus sensitive (a property that we will define later) and if entities may hold arbitrarily large liabilities, then the total risk can even be made arbitrarily small. 

\begin{rem}\label{consol}
Capital regulation, such as Solvency II, may require the computation of a group SCR on the basis of consolidated data. For example, if a full internal model for the group exists, a solvency balance sheet on the group level is required (see Title III: Supervision of Insurance and Reinsurance Undertakings in a Group, \ci{EC}). Although the entities of the group are legally separate, regulation takes place at the level of the group; this is often referred to as the \emph{legal entity fiction} of group regulation. 

This artificial entity is unproblematic, if capital regulation is based on a coherent risk measure. In this case, the sum of the individual SCRs is always at least as large as the group SCR computed from consolidated data. As a consequence,  coherent risk measures provide corporate groups with an intrinsic incentive to base their analysis on consolidated balance sheets. In contrast, V@R-type risk measures create the opposite incentives.  Solvency II is an example where the SCR is computed by such a risk measure, i.e. V@R. V@R-type risk measures result in a misalignment of regulatory objectives and rational behavior of corporations. They entice insurance firms to explore alternative network structures, not classified as groups. 

Sophisticated corporations could e.g. set up a complex network of companies owned by multiple other entities that are partially located outside the EEA.  If well designed, regulation of parts of these networks might be out of reach 
for EEA supervisory authorities. If legal obligations regarding consolidation (see \ci{EC_Cons}) are properly reflected in the construction of the structures, an application of Article 262 (2) of Directive 2009/138/EC (\ci{EC}) might not be viable.  

But the problems of V@R-type risk measures reach beyond centrally controlled networks. An analysis of market equilibria in the special case of RV@R in \ci{embrechtsliuwang2016} indicates that central governance is not required to produce suitable allocations.\footnote{Future research might attempt to investigate market equilibria for the classes of risk measures that we consider in this paper, i.e. V@R-type distortion risk measures and general V@R-type risk measures.} In summary, the belief that current Solvency II group regulation suffices for general corporate networks does not seem to be well justified.  
\end{rem}

\subsection{Risk sharing for V@R-type distortion risk measures}

To simplify the technical arguments, we work on an atomless probability space $(\Omega, \fil, P)$. This means that a uniform random variable is defined on this space. 

We now consider the optimal risk sharing problem
\begin{equation}\label{eq:riskshare}
 \square_{i=1}^n \rho^i \; (X) \; := \;  \inf \left\{\sum_{i=1}^n \rho^i(X^i): \; \sum_{i=1}^n X^i = X, \quad X^1, X^2, \dots, X^n \in L^\infty  \right\}.
 \end{equation}
The following theorem provides an upper bound to the solution and an allocation that attains this bound. 

\begin{thm}\label{thm:main1}

Let $X\in L^\infty$ and $n\in \bbn$. By $g^1, g^2, \dots, g^n$ we denote left-continuous distortion functions with parameters $\alpha_1, \alpha_2, \dots, \alpha_n \in [0,1)$ and define $d=\sum _{i=1}^n \alpha_i$. We set $\rho^i = \rho^{g_i}$, i.e. $\rho^i$ is the distortion risk measure associated with the distortion function $g^i$, $i=1,2, \dots, n$.  Define the left-continuous functions $$f= \min\{\widehat {g^1}, \widehat {g^2}, \dots, \widehat {g ^n}\}, \quad \quad g(x) =  
\left\{
\begin{array}{ll}
0 , & 0 \leq x \leq d\wedge 1, \\
f( x - d), & d \wedge 1 < x \leq 1
\end{array}
\right.
$$
Note that $g \equiv 0$, if $d\geq 1$. In particular, $g$ is not necessarily a distortion function with $g(1)=1$. We set $V@R_\lambda := V@R_1 = \essinf$ for $\lambda \geq 1$.
\begin{enumerate}
\item There exist $X^1, X^2, \dots, X^n \in L^\infty$ such that $\sum_{i=1}^n X^i = X$ and 
$$ \sum_{i=1}^n \rho^i (X^i)  = \int _{[0,1]} V@R_\lambda (X - \essinf X) g (d\lambda)  + \essinf X.$$
If $d\geq 1$, this equation can be simplified and we obtain 
$$   \sum_{i=1}^n \rho^i (X^i)  = \essinf X .  $$
\item  
The allocation $(X^i)_{i=1,2, \dots, n}$ can be constructed as follows.
Let $$Y \; : = \;   X - \essinf X\;  \geq \; 0.$$ There exists a random variable $U$, uniformly distributed on $[0,1]$, such that $Y= V@R_U (Y)$.
\noindent 
For $i=1,2, \dots, n$, we set
$$ r_i(\lambda ) = 
\left\{
\begin{array}{ll}
1, & i = \inf \{j: \hat g_j (1- \lambda) = f (1-\lambda) \}, \\
0,  & \mbox{else,}
\end{array}
\right.
 $$
 $(\lambda \in [0,1])$ and $R_i(y) = \int _0^y r_i(\lambda) d \lambda$. We define $ \tilde Y = Y \cdot \1_{\{  U \geq d \}}$ and $\tilde Xî = R_i (\tilde Y)$. 
For $i=1,2, \dots, n$, we set
\begin{eqnarray}\label{allocation}
X^i & = &  Y \cdot \1 _{\{\sum_{l = 1}^{i-1 }  \alpha_l \; \leq \;  U  \; <  \; \sum_{l = 1}^{i} \alpha_l \}} \;  + \;  \tilde X ^i \; +  \; \frac{\essinf X }{n}
\end{eqnarray}
If $d\geq 1$, this equation can be simplified and we obtain
\begin{eqnarray}\label{xisimple}
X^i & = & Y \cdot \1 _{\{\sum_{l = 1}^{i-1 } \alpha_l  \; \leq \;  U  \; <  \; \sum_{l = 1}^{i} \alpha_l \}} \;   + \;  \frac{\essinf X}{n}
\end{eqnarray}
\end{enumerate}

\end{thm}

\begin{proof}
See Section \ref{proof:main1}.
\end{proof}

\begin{corollary}\label{cor1}
Suppose that the conditions of Theorem \ref{thm:main1} hold. 
The solution to the optimal risk sharing problem \eqref{eq:riskshare} is bounded by
$$ \square_{i=1}^n \rho^i \; (X) \leq   \int _{[0,1]} V@R_\lambda (X - \essinf X) g (d\lambda)  + \essinf X. $$ In particular, if $d\geq 1$, this bound is equal to the total risk of the best case scenario $\essinf X$ of $X$ evaluated by an arbitrary normalized risk measure, i.e.
$$ \square_{i=1}^n \rho^i \; (X) \leq   \essinf X  .$$
\end{corollary}

\begin{proof}
See Section \ref{proof:cor1}.
\end{proof}

Let us now specify additional assumptions such that the upper bound of Corollary \ref{cor1} is at the same time a lower bound and thus equal to the value of the optimal risk sharing problem. 

\begin{thm}\label{thm:main2}
Suppose that the conditions of Theorem \ref{thm:main1} hold. In addition, assume that $d<1$ and $g^i(1-d+\alpha_i) =1$ for $i=1,2, \dots, n$, and that the active parts of the distortion functions $g^1, g^2, \dots, g^n$ are concave. Then the allocation defined in eq. \eqref{allocation} provides a solution to the optimal risk sharing problem \eqref{eq:riskshare}  and 
$$ \square_{i=1}^n \rho^i \; (X) =   \int _{[0,1]} V@R_\lambda (X ) g (d\lambda)  . $$
\end{thm}

\begin{proof}
See Section \ref{proof:main2}.
\end{proof}

\begin{rem}
Distortion risk measures with concave active parts and parameter $\alpha >0 $ were discussed in Example 3.3 in \ci{WBT15} in the context of robust modifications of coherent risk measures. 
\end{rem}

Theorem \ref{thm:main1}, Corollary \ref{cor1} and Theorem \ref{thm:main2} provide an important perspective on capital regulation based on V@R-type distortion risk measures. They show that (if risk is measured by a normalized risk measure and the network consists of sufficiently many entities)  the total capital requirement can be made equal to the capital requirement of the best case scenario of the network, i.e.  $\essinf X = - {\rm esssup} \; E_1$. This quantity is an upper bound to the solution of the optimal risk sharing problem. Downside risk can thus completely be hidden within corporate network structures. V@R is a special case of a V@R-type distortion risk measure, and our observations apply to Solvency II. In contrast, they do not apply to the Swiss Solvency Test that uses the coherent risk measure AV@R as the basis for capital regulation. 

Observe that the allocation $(X^i)$ in equation \eqref{xisimple} is bounded from below by $\frac {\essinf X}{n}$ and from above by ${\rm esssup~} Y +  \frac {\essinf X}{n}$. We could thus restrict the admissible allocations to those that are bounded by suitable fixed constants and still obtain the results stated above. If no bounds are imposed, the situation can even be more serious from the point of view of capital regulation. We will show  in Example \ref{RV@Rbd}, Theorem \ref{thm:small}, and Remark \ref{domains} that the total capital requirement of the network can be further reduced, if no bounds are imposed on admissible profits and losses of network entities. In these cases, risk sharing can be used to make the total risk $\sum_{i=1}^n \rho^i (X^i)$ arbitrarily small for appropriately chosen allocations, i.e. smaller than $-m$ for any $m \in \bbn$. In Example \ref{RV@Rbd} and Theorem \ref{thm:small} very large losses of one entity may occur due to transfer agreements with an other entity that experiences large profits in the corresponding scenarios. Remark \ref{domains}, in contrast, parallels the results of  Theorem \ref{thm:main1} and Corollary \ref{cor1}, but in the case where $X$ itself is unbounded. 

\begin{ex}\label{RV@Rbd}
Let $(\Omega, \fil, P)$ be a probability space without atoms. Consider a corporate network of $n=2$ entities with risk measures $\rho^1 = \rho^2 = RV@R_{\frac{1}{4}, \frac{3}{4}}$. We will show\footnote{General results on inf-convolutions of RV@R are given in \ci{embrechtsliuwang2016} that imply the results given in the example.} that $\square_{i=1}^2 \rho^i \; (0) = - \infty $.
To this end, let $A_1, A_2 \subseteq \Omega$ be a partition of $\Omega$ such that $P(A_1) = \frac 1 8$, $P (A_2) = \frac 7 8 $. Let $m\in \bbn$ be arbitrary, and set $X^1 := 6m \cdot \1_{A_1}$, $X^2 := - 6m \cdot \1_{A_1}$. Then $X^1 + X^2 = 0$, $\rho^1(X^1) = 0$, $\rho^2 (X^2) = -6m\cdot \frac 1 8 \cdot \frac 4 3 = -m$. Thus, $\square_{i=1}^2 \rho^i \; (0)  \leq -m$ for any $m\in \bbn$. 
\end{ex}

We now provide a theorem that characterizes the situation of the previous example on a general level.

\begin{thm} \label{thm:small}
Suppose that the conditions of Theorem \ref{thm:main1} hold and assume that there exists $i\in \{ 1,2, \dots, n   \}$ such that $g^i( 1- d + \alpha_i) <1$. Then 
$$  \square_{i=1}^n \rho^i \; (X)\; = \; - \infty . $$
\end{thm}

\begin{proof}
See Section \ref{proof:small}.
\end{proof}

From a regulatory point of view, under the conditions of the last theorem, capital regulation can completely be circumvented in corporate networks: the total downside risk measurements are not bounded from below anymore. One should, however, note that network allocations with arbitrarily small total risk are associated with arbitrarily large losses and profits of some of the entities of the network. For insurance networks, the implementation of the required transfer agreements might not be realistic, if the admissible leverage is bounded for each entity. This means that in practice Theorem \ref{thm:small} is less relevant for capital regulation than Theorem \ref{thm:main1}, Corollary \ref{cor1} and Theorem \ref{thm:main2}. It stresses, however, potential problems that might occur if significant leverage and V@R-type distortion risk measures are used together.

\begin{rem}\label{domains}
We have been considering the risk sharing problem for bounded $X\in L^\infty$, but this restriction is not necessary. Suppose now that $X$ is an arbitrary random variable. If $X$ is bounded from below, i.e. $\essinf X > - \infty$, it is not difficult to verify that the results of Theorem \ref{thm:main1} are still valid. This is due to the fact that in Theorem \ref{thm:main1} total losses $X$ beyond $V@R_d(X)$ are allocated to the positions $(X^i)_{i=1,2, \dots, n}$ such that they do not influence the risk measurements $(\rho^i(X^i))_{i=1,2, \dots,n}$. 

Next, let us assume that $X$ is not bounded from below anymore, i.e. $\essinf X = - \infty$. Consider the case: $n$ large enough and $d\geq 1$. Then $\essinf (X \vee (-k)) = - k $ for $k\geq0$. By the monotonicity of risk measures, $$ \square_{i=1}^n \rho^i  (X) \; \leq \;   \square_{i=1}^n \rho^i  (X \vee (-k)) \; = \; -k \;  \stackrel{k\to \infty}{\longrightarrow}    \;  - \infty .$$
We thus obtain a result that is analogous to the situation of Theorem \ref{thm:small}: the total downside risk measurement is not bounded from below, if the best case is unbounded.

\end{rem}

\subsection{Risk sharing for V@R-type risk measures}\label{sec:V@Rtype}

So far, we have been focussing on distortion risk measures. For $d<1$ we do, indeed, need this particular structure to compute the exact total risk of the allocation defined in eq. \eqref{allocation}. This allocation provides firstly an upper bound for the total risk of the optimal risk sharing problem and secondly a solution in the case of concave active parts. In the case of $d\geq 1$ the allocation defined in eq. \eqref{xisimple} provides a bound, and it turns out that this result is not limited to the family of distortion risk measures. The next theorem provides a precise statement. In addition, we will also be able to generalize Theorem \ref{thm:small} beyond the case of distortion risk measures. 

\begin{defi}
A distribution-based risk measure $\rho: L^\infty \to \bbr$ is a \emph{V@R-type risk measure} with parameter $\alpha>0$, if 
$$\rho (X)\;  = \; \rho \left(X \cdot {\bf 1}_{\{V@R_\alpha (X)\geq X\}}  + V@R_\alpha (X) \cdot {\bf 1}_{\{V@R_\alpha (X) < X\}} \right)\quad\quad\quad  (X\in L^\infty) .$$
\end{defi}

\begin{rem}
Obviously, any V@R-type distortion risk measure with parameter $\alpha >0 $ is a V@R-type risk measure with the same parameter. This follows immediately from Remark \ref{rem:repre}. Moreover, any V@R-type risk measure with parameter $\alpha >0 $ is dominated by $\rho(0) + V@R_\alpha $, since $\rho$ is monotone and $X \cdot {\bf 1}_{\{V@R_\alpha (X)\geq X\}}  + V@R_\alpha (X) \cdot {\bf 1}_{\{V@R_\alpha (X) < X\}} \leq V@R_\alpha (X)$ for $X\in L^\infty$.
\end{rem}

\begin{thm}\label{thm:main3}
Let $X\in L^\infty$, $\rho^1, \rho^2, \dots, \rho^n$ be V@R-type risk measures with parameters $\alpha_1, \alpha_2, \dots, \alpha_n$, and the allocation $(X^i)_{i=1,2, \dots, n }$  be given according to equation \eqref{xisimple}. If $d=\sum_{i=1}^n \alpha_i \geq 1$, then  
$$ \square_{i=1}^n \rho^i \; (X) \leq \; \sum_{i=1}^n \rho^i (X^i)  \; =  \;  \sum_{i=1}^n \rho^i(0)  \; + \;  \essinf X.  $$
 In particular, if the risk measures $\rho^i$ are normalized, i.e. $\rho^i(0) =0$, $i=1,2,\dots, n$, then the minimal total risk is bounded by the risk of the best case scenario $\essinf X$ of $X$ evaluated by an arbitrary normalized risk measure. 
\end{thm}

\begin{proof}
See Section \ref{proof:main3}.
\end{proof}

\begin{rem}
Arguments analogous to those in Remark~\ref{domains} show that the result of Theorem \ref{thm:main3} is not limited to positions in the space $L^\infty$, but holds for larger spaces of random variables. If $\essinf X = - \infty $, the total risk measurement can be made arbitrarily small for suitable risk sharing allocations. 
\end{rem}

We finally show how Theorem \ref{thm:small} may be generalized beyond the case of distortion risk measures. 

\begin{defi}
A distribution-based risk measure $\rho: L^\infty \to \bbr$ is \emph{surplus sensitive at level $\alpha >0$}, if
$$\rho (X)> \rho (X - m \cdot \1 _{\{  V@R_{1-\alpha} (X) \geq X  \}} ) =: h_x(m)$$
for any $m>0$.

If, in addition, $h_X(m) \to - \infty$ as $m\to \infty$, then $\rho$ is \emph{strongly surplus sensitive at level $\alpha >0$}.
\end{defi}

\begin{rem}
\begin{enumerate}
\item A risk measure that is surplus sensitive is not necessarily strongly surplus sensitive at the same level. An example is the entropic risk measure. We consider the special case $\rho(X) = \log E(e^X)$. Setting $\alpha = 1/5$, we compute $h_X$ for a random variable $X$ with $P(X= \log 10) = 1/10$ and $P(X=0) = 9/10$. Then $h_X(m) = \log E(e^{X-m\cdot \1 _{V@R_{4/5} (X) \geq X}}) = \log (\frac 1 {10} e^{\log 10} + \frac 9 {10} e^{-m}) = \log (1 + \frac {9}{10} {e^{-m}} ) \stackrel{m\to \infty} {\longrightarrow} 0 > - \infty$.
\item 
In contrast, any distortion risk measure $\rho^g$ with distortion function $g$ such that $g(x)< 1$ for $x<1$ is strongly surplus sensitive at any level $d>0$.  Another example are expectiles with acceptance set 
$ \left\{  X\in L^\infty : \quad \frac {E(X^-)}{E(X^+)} \geq \gamma  \right\}$
for $\gamma >0$.
\end{enumerate}
\end{rem}

\begin{thm}\label{thm:smallsimple}
Let $X\in L^\infty$ and $\rho^1, \rho^2, \dots, \rho^n$ be V@R-type risk measures with parameters $\alpha_1, \alpha_2, \dots, \alpha_n$. Set $d = \sum_{i=1}^n \alpha_i$.
If $\rho^{n+1}$ is strongly surplus sensitive at level $d$, then 
$$\square_{i=1}^{n+1} \rho^i (X) = - \infty  . $$
\end{thm}

\begin{proof}
See Section \ref{sec:smallsimple}.
\end{proof}

\begin{rem}
Theorem \ref{thm:small} can be seen as a corollary to Theorem \ref{thm:smallsimple}. But the direct proof of Theorem \ref{thm:small} in Section \ref{proof:small} computes in addition explicitly the exact total risk of the allocations that were considered to bound the value of the risk sharing problem \eqref{eq:riskshare} from above. 

In contrast to Theorem \ref{thm:main3}, the practical relevance of Theorem \ref{thm:smallsimple} might be limited by the fact that the required allocations are associated with arbitrarily large losses and profits, i.e. leverage is unbounded. Nevertheless, like Theorem \ref{thm:small}, it highlights the problems that arise when V@R-type risk measures and leverage are combined. 
\end{rem}

\section{Conclusion}\label{sec:concl}

Non-coherent risk measures have frequently been criticized during the last twenty years. The present paper presents another challenge to capital regulation with non-coherent risk measures of V@R-type: sophisticated firms might be able to hide their downside risk within corporate networks by designing suitable intra-network transfers. In this paper, these network transfers are specified as derivatives on the stochastic balance sheet of the network. Future research needs to express these as contingent claims on tradeable securities (or, at least, on quantities that cannot easily be manipulated by the managers of the firm). 

The feasibility of the discovered capital reduction strategy relies on the fact that capital requirements are charged separately for each entity of the network. Intra-network transfers would not reduce the total capital requirement, if the latter was defined on the basis of a consolidated balance sheet. Solvency II essentially requires such an approach for corporate groups, but it seems to be difficult to define a universal legal framework that reaches beyond the EEA. At the same time, results by \ci{embrechtsliuwang2016} indicate that central control is not necessary. However, even if a consolidated approach to multinational corporate networks could be successfully implemented globally, serious problems would remain. Networks consist of legally separate firms with limited liability. If capital requirements are computed on the basis of a consolidated balance sheet, groups could exploit multiple limited liability options associated with their subentities via appropriate strategic defaults at the expense of third parties, optimizing their own gains with prior intra-group transfers. Consolidated balance sheets do not reflect these possibilities. Our argument indicates that consolidated balance sheets should be accompanied by joint liability of the group and \emph{not} be combined with limited liability of subentities. 

In summary, from a regulatory point of view V@R-type risk measures and corporate networks of legally separate entities are not fully compatible. If one recalls other well-known deficiencies of non-coherent risk measures, V@R-type risk measures do not seem to be an optimal ingredient to the regulation of insurance companies.

\vfill

\pagebreak

\appendix

\section{Proofs}\label{sec:proofs}

\subsection{Proof of Theorem \ref{thm:main1}}\label{proof:main1}

The methodology of the proof is primarily based on \ci{Dhaene2012}. The allocation of the downside risk to different positions such that it `can be swept under the carpet' relies on the same strategy as in \ci{embrechtsliuwang2016}. The allocation of the remaining part is inspired by Proposition 5 in \ci{embrechtsliuwang2016} which focuses on comonotone allocations and is originally due to \ci{cuietal}.

\begin{proof}

The existence of a random variable $U$, uniformly distributed on $[0, 1]$, such that $Y = V @R_U (Y )$ is a version of the inverse transform method, see e.g. Section 2.2.1 in \ci{glasserman}.

\noindent We note that $\sum_{i=1}^n r_i(\lambda) = 1$ and 
\begin{equation}\label{Eq_RG}
\sum_{i=1}^n r_i(\lambda) \hat g_i (1- \lambda) = f (1 -\lambda) 
\end{equation}
for $\lambda \in [0,1]$. 

\noindent First, we observe that 

\begin{eqnarray*}
\sum_{i=1} ^n  X^i & = & Y \cdot \sum_{i=1}^n \1_{\{ \sum_{l=1}^{i-1} \alpha_l \leq U < \sum _{l=1}^i \alpha_l   \}}  \; +  \; \sum_{i=1}^n \tilde X^i  \; + \;  \essinf X \\
& = & Y \cdot \1_{\{U<d \}}  \; + \;  \int _0 ^{\tilde Y} \sum_{i=1}^n r_i(\lambda) d \lambda  \; + \;  \essinf X  \\
& = &   Y \cdot \1_{\{U<d \}}  \; + \;  Y \cdot \1_{\{U\geq d \}}  \; + \;    \essinf X \;  = \; X . 
\end{eqnarray*}

\noindent For any $i=1 ,2, \dots, n$, 
$$
 Y \cdot \1_{\{ \sum_{l=1}^{i-1} \alpha_l \leq U < \sum _{l=1}^i \alpha_l   \}}  \geq V@R_d (Y)\geq   \tilde X ^i , \quad   P\left( \left\{ \sum_{l=1}^{i-1} \alpha_l \leq U < \sum _{l=1}^i \alpha_l   \right\}   \right) = \alpha_i. 
$$
Thus, for $\lambda \geq \alpha_i$ we obtain that
\begin{equation}\label{Eq_VaR1}
V@R_\lambda (X^i) \; = \;  V@R_{\lambda - \alpha_i} (\tilde X ^i ) + \frac{\essinf X}{n} .
\end{equation}
Similarly, for $\lambda \geq d$ we have 
\begin{equation}\label{Eq_VaR2}
V@R _\lambda (Y) = V@R_{\lambda- d} (Y \cdot \1 _{\{ U \geq d \}}),
\end{equation}
and for $\lambda \geq 1-d$,
\begin{equation}\label{Eq_VaR3}
V@R _\lambda (Y \cdot \1 _{\{ U \geq d \}}) = 0.
\end{equation}

\noindent Next, we recall some facts from \ci{Dhaene2012}. Let $Z \geq 0$ be a non-negative random variable. We denote the distribution function of $Z$ by $F_Z$. The caglad quantile function of $Z$ is given by
$$
[0,1] \to [0, \infty], \quad \lambda \mapsto F_Z^{-1} (\lambda) = \inf \{z: F_Z(z) \geq \lambda \},
$$
its cadlag value at risk function by
$$
[0,1] \to [0, \infty], \quad \lambda \mapsto V@R_\lambda (Z) = F^{-1}_Z(1-\lambda).
$$
For any left-continuous distortion function $g$ we obtain from Definition 3 in  \ci{Dhaene2012}:
\begin{equation}\label{Eq_Def_Dist}
\rho^g (Z)  \; =  \; \int_0 ^\infty g (1 - F_Z(s)) ds.
\end{equation}
Recall from Remark \ref{rem:repre} that $\rho^g (Z) = \int_0^1  V@R_\lambda (Z) g (d \lambda)$, even if $Z$ is not non-negative.

\noindent Equation (2) in  \ci{Dhaene2012} implies
\begin{equation}\label{Eq_Inequ}
s < V@R_\lambda (Z) \quad \Leftrightarrow \quad F_Z(s) < 1 - \lambda. 
\end{equation}
\noindent Finally, we obtain 
\begin{eqnarray*}
\sum_{i=1} ^n \rho^{i} (X^i)  &   = & \sum_{i=1}^n  \int_0^1 V@R_\lambda (X^i) d g^i (\lambda)  
 \\
  & \stackrel{eq.~\eqref{Eq_VaR1} }{=} & \essinf X + \sum_{i=1 }^n \int_{[\alpha_i ,1]}    V@R _{\lambda- \alpha_i}    (\tilde X^i) d g ^{i} (\lambda)\\
  & = &  \essinf X + \sum_{i=1 }^n \int_{[0, 1- \alpha_i ]}    V@R _{\lambda}    (\tilde X^i) d \hat g ^{i} (\lambda)\\
  & = &  \essinf X + \sum_{i=1 }^n \int_0 ^1    V@R _{\lambda}    (\tilde X^i) d \hat g ^{i} (\lambda) \quad \quad \mbox{(since $\hat g ^i(\lambda ) = 1$ for $\lambda\geq 1 - \alpha_i$)}\\
  & = &  \essinf X + \sum_{i=1 }^n \int_0 ^1    V@R _{\lambda}    (R_i(\tilde Y)) d \hat g ^{i} (\lambda) \\
  & = &  \essinf X + \sum_{i=1 }^n \int_0 ^1  R_i[   V@R _{\lambda}    (\tilde Y) ] d \hat g ^{i} (\lambda) \quad \quad \mbox{(since $R_i$ monotone increasing)}
  \\
  & = &  \essinf X + \sum_{i=1 }^n \int_0 ^1 \int_0^\infty \1_{[0, V@R _{\lambda}    (\tilde Y) )}  (s) r_i(s) ds d \hat g ^{i} (\lambda) \\
  & = &  \essinf X + \sum_{i=1 }^n \int_0^\infty \int_0 ^1  \1_{[0, V@R _{\lambda}    (\tilde Y) )}  (s)  d \hat g ^{i} (\lambda) r_i(s) ds \quad \mbox{(by Fubini's theorem)} \\
  & \stackrel{eq.~\eqref{Eq_Inequ} }{=} &  \essinf X + \sum_{i=1 }^n \int_0^\infty  g_i[ (1- F_{\tilde Y} (s)) - ] r_i(s) ds 
   \\
  & {=} &  \essinf X +  \int_0^\infty    \sum_{i=1 }^n g_i (1- F_{\tilde Y} (s))   r_i(s) ds \quad \mbox{(since $g_i$ left-continuous)}   
  \\
  & \stackrel{eq.~\eqref{Eq_RG} }{=}  &  \essinf X +  \int_0^\infty    f (1- F_{\tilde Y} (s))    ds 
   \\
  & \stackrel{eq.~\eqref{Eq_Def_Dist} \; \&\; Rem.~\ref{rem:repre}  }{=}  &  \essinf X +  \int_0^1   V@R_\lambda(Y \cdot \1 _{\{ U \geq d\}})  d f (\lambda)
   \\
  & \stackrel{eq.~\eqref{Eq_VaR3} }{=}  &  \essinf X +  \int_0^{1-d}   V@R_\lambda(Y \cdot \1 _{\{ U \geq d\}})  d f (\lambda)
     \\
  & \stackrel{eq.~\eqref{Eq_VaR2} }{=}  &  \essinf X +  \int_0^{1}   V@R_\lambda(Y )  d g (\lambda)
     \\
  &=  &  \essinf X +  \int_0^{1}   V@R_\lambda(X - \essinf X  )  d g (\lambda)
   \end{eqnarray*}

\end{proof}

\subsection{Proof of Corollary \ref{cor1}}\label{proof:cor1}

\begin{proof}
Obviously,
$$  \square_{i=1}^n \rho^i \; (X)\; \leq \; \sum_{i=1}^n \rho^i(X^i)$$
for the allocation defined in equation \eqref{allocation}. Thus, both claims follow from Theorem \ref{thm:main1}.
\end{proof}

\subsection{Proof of Theorem \ref{thm:main2}}\label{proof:main2}

\begin{proof}

Note first that $1=g^i (1-d+\alpha_i) = \widehat{g^i} (1-d)$  for $i=1,2, \dots,n$, thus $g(1)= f(1-d) = \min\{ \widehat{g^1} (1 -d ), \widehat{g^2} (1-d), \dots, \widehat{g^n} (1-d)  \} = 1$.

Observe that the inequality ``$\leq$'' follows from Corollary \ref{cor1}, since $g(1) = 1$ and V@R is cash-invariant. We present a proof for the inequality ``$\geq$''.\fn{A preliminary draft version of this proof can be found in the B.Sc. thesis \ci{Agirman} that was supervised by the author of this paper.} We prove the statement by induction over the number $n$ of distortion functions $g^1, g^2, \dots, g^n$.

Consider first two distortion functions $g^i$, $i=1,2$, i.e. $n=2$. In this case, $d = \alpha_1 + \alpha _2$.
Given any $X^1, X^2 \in L^\infty$ with $X^1 + X^2 =X$ we construct $Y^1, Y^2 \in L^\infty$ such that 
\begin{eqnarray*}
\rho^1 (X^1) + \rho^2 (X^2) 
\; &
\stackrel{(\theequation)\label{a}}
{=}  &
\;
 \rho^{\widehat {g^1}} (Y^1) +\rho^{\widehat {g^2}}  (Y^2) 
\;
\stackrel{\refstepcounter{equation}(\theequation)\label{b}}
{\geq}  
\;
\rho^f (Y^1) + \rho^f (Y^2) \\
\;
&\stackrel{\refstepcounter{equation}(\theequation)\label{c}}
{\geq}   &
\; \quad
\rho^f  (Y^1 + Y^2)  
\;\quad 
\stackrel{\refstepcounter{equation}(\theequation)\label{d}}
{\geq}  
\; \;
\int_{[0,1]} V@R_\lambda (X^1 + X^2) g(d\lambda) .   
\end{eqnarray*}
Observe that by Theorem 6 in \ci{Dhaene2012} we have that $\int V@R_\lambda (X^1 + X^2) g(d\lambda) = \rho^g (X^1 + X^2)$, since $g(1)=1$. 

\noindent We will first specify $Y^1, Y^2$ and then verify the inequalities \eqref{a} -- \eqref{d}.

To begin with, observe that $\tilde X^1 = X^1 - \essinf X^1 \geq 0$, $\tilde X^2 = X^2  -\essinf X^2 \geq 0$. If $\rho^1(\tilde X^1) + \rho^2(\tilde X^2)  \geq \rho^g (\tilde X^1 +  \tilde X^2)$, we add  $ \essinf X^1 + \essinf X^2$ to both sides in order to obtain by cash-invariance that $\rho^1( X^1) + \rho^2( X^2)  \geq \rho^g ( X^1 +   X^2)$. We may thus assume w.l.og. that $X^1, X^2 \geq 0$ and will do so henceforth. Let $U^1$ and $U^2$ be random variables, uniformly distributed on $[0,1]$ such that $X^1 = V@R_{U^1} ( X^1)$ and  $X^2 = V@R_{U^2} ( X^2)$.

Set $Y^1 = X^1 \cdot \1_{\{\alpha_1  \leq U^1  \}} $, $Y^2 = X^2 \cdot \1_{\{ \alpha_2 \leq U^2  \}} $.
Then
$$V@R_\lambda(Y^i) =
\left\{
\begin{array}{ll}
V@R_{\alpha_i + \lambda} (X^i), & \lambda < 1 - \alpha_i ,\\
0, &  1 - \alpha_i \leq   \lambda.
\end{array}
\right.
$$
This implies by Theorem 6 in \ci{Dhaene2012} that 
\begin{eqnarray*}
&& \rho^1 (X^1) + \rho^2 (X^2) = \int V@R_\lambda (X^1) d g_1(\lambda) + \int V@R _ \lambda (X^2) d g _2 (\lambda) \\
& = & \int V@R_{\alpha_1 + \lambda} (X^1) d \hat g_1 (\lambda) + \int V@R_{\alpha_2 + \lambda} (X^2) d \hat g_2 (\lambda) = \int V@R_{\lambda} (Y^1) d \hat g_1 (\lambda) + \int V@R_{ \lambda} (Y^2) d \hat g_2 (\lambda)\\
& = & \rho^{\widehat {g^1}} (Y^1) + \rho^{\widehat {g^2}} (Y^2), \mbox{~i.e.~equation~} \eqref{a}.
\end{eqnarray*}
Now observe that by definition $\widehat {g^i}\geq f$, $i=1,2$. Thus, $ \rho^{\widehat {g^1}} (Y^1) + \rho^{\widehat {g^2}} (Y^2) \geq \rho^f (Y^1) + \rho^f (Y^2)$, i.e. inequality \eqref{b}. Since $f$ is concave, it follows moreover that $\rho^f(Y^1) + \rho^f (Y^2) \geq \rho^f(Y^1 + Y^2)$, i.e. inequality \eqref{c}.

Let $A_i := \{ \alpha_i > U^i  \}$, $i=1,2$. Then $P(A_i) = \alpha_i$, $i=1,2$. Observe that $Y^i=X^i$ on the complement $A_i^c$ of $A_i$, $i=1,2$. For $x\in \bbr$ we get
\begin{eqnarray*}
&& P\{ Y^1 + Y^2 >x\} \geq P(\{ Y^1 + Y^2 >x \} \cap (A_1 \cup A_2)^c)\\
& = & P(\{ X^1 + X^2 >x \} \cap (A_1 \cup A_2)^c) \geq P\{ X^1 + X^2 >x \} - P (A_1 \cup A_2)\\
& \geq & P\{ X^1 + X^2 >x \}  -(\alpha_1 + \alpha_2) = P\{ X^1 + X^2 >x \}  - d.
\end{eqnarray*}
Since $P\{ Y^1 + Y^2 >x\} \geq 0$, we get 
$$P\{ Y^1 + Y^2 >x\} \geq (P\{ X^1 + X^2 >x \}  - d) \vee 0  =  (P\{ X^1 + X^2 >x \}  \vee d ) -d . $$
We have $Y^1 + Y^2 \geq 0$ by construction, thus
\begin{eqnarray*}
&& \rho^f (Y^1 + Y^2) = \int _0^\infty f(P\{ Y^1 + Y^2 >x \}) dx \geq \int_0^\infty f([P\{ X^1 + X^2 >x \}  \vee d] -d) dx\\
&\stackrel{\refstepcounter{equation}(\theequation)\label{e}}{=}& \int _0^\infty g (P\{ X^1 + X^2 >x \}) dx \stackrel{\refstepcounter{equation}(\theequation)\label{f}\refstepcounter{equation}}{=} \rho^g (X^1 + X^2),
\end{eqnarray*}
where we observe for \eqref{e} that $f((y \vee d) -d) = g(y)$ and for \eqref{f} that $X^1 + X^2 \geq 0 $ by assumption. This shows \eqref{d}.

Next, we show that the claim holds for $n+1$ distortion functions, if it holds for up to $n$ distortion functions. Assume that the induction hypothesis is true, and let $g^1, g^2, \dots, g^{n+1}$ be distortion functions with parameters $\alpha_1, \alpha_2, \dots, \alpha_{n+1} \in [0,1)$ and concave active parts. In this case, $d= \sum_{i=1}^{n+1} \alpha_i$. The corresponding distortion risk measures are again denoted by $\rho^1 , \rho^2, \dots, \rho^{n+1}$. Define 
\begin{align*}
f^{(j)} = \min \{   \widehat {g^1}, \widehat {g^2}, \dots, \widehat {g^j} \}, \quad 
d^{(j)} = \sum_{i=1} ^j \alpha_i, \quad 
g^{(j)} (x) 
= &
\left\{
\begin{array}{ll}
0, & 0 \leq x \leq d^{(j)}, \\
f^{(j)} (x - d^{(j)}), & d^{(j)} < x \leq 1 
\end{array}
\right.
\\
&\quad \quad\quad\quad\quad\quad\quad\quad (j=n, n+1)
\end{align*}
Let $X^1, X^2, \dots, X^{n+1} \in L^\infty$ such that $\sum_{i=1}^{n+1} X^i = X$.
Set 
$$
h(x)
= 
\left\{
\begin{array}{ll}
0, & 0 \leq x \leq d, \\
f(x-d), & d < x \leq 1,
\end{array}
\right.
$$
with $f = \min \{ \widehat{ g^{(n)}}, \widehat {g^{n+1}}  \}$. Then, using the induction hypothesis twice, we get that
\begin{equation}\label{eq:lb}
\rho^h(X) \leq  \rho ^{g^{(n)}} (\sum_{i=1}^n X^i ) + \rho^{n+1} (X^{n+1}) \leq \sum_{i=1}^{n+1} \rho^i (X^i).
\end{equation}
Finally note that $d=d^{(n)} + \alpha_{n+1} = \sum_{i=1}^{n+1} \alpha_i$ and $f = \min \{ \widehat {g^{(n)}}, \widehat {g^{n+1}}  \} =  \min \{ \widehat {g^1}, \widehat {g^2}, \dots   ,  \widehat {g^{n+1}}  \} $, thus $h= g^{(n+1)} $. By Theorem 6 in \ci{Dhaene2012} we finally rewrite the left-hand side of equation \eqref{eq:lb} as
$$\rho^h(X)  = \rho^{g^{(n+1)}}(X)  = \int_{[0,1]} V@R_\lambda (X) g^{(n+1)} (d\lambda).$$
This proves the claim. 
\end{proof}

\subsection{Proof of Theorem \ref{thm:small}}\label{proof:small}

\begin{proof}
Due to the cash-invariance of risk measures, we may w.l.o.g suppose that $X\geq 0$. Let $U$ a random variable, uniformly distributed on $[0,1]$ such that $V@R_U(X) = X$. Renumbering the distortion functions and risk measures, we assume that $g^1(1-d+\alpha_1) <1$. For $m \in \bbn$ we set
\begin{eqnarray*}
X^{1,m} & := &  X \cdot \1_{\{U < \alpha_1\}}   + X \cdot \1 _{\{ d \leq U \}} - m \cdot \1 _{\{\alpha_1 \leq U < d\}} , \\
X^{i,m} &  := &  (X+m) \cdot \1_{\{\sum_{l=1}^{i-1} \alpha_l \leq U <  \sum_{l=1}^{i} \alpha_l \}} \quad \quad\quad (i=2,3, \dots,n).
\end{eqnarray*}
Obviously, by construction $\sum_{i=1}^n X^{i,m} =X$.

Since $g^i$ is a distortion function with parameter $\alpha_i>0$ and $X\geq 0$, thus $X+m >0$, we get $$\rho^i(X^{i,m}) =0,  \quad \quad i=2,3, \dots, n.$$
We compute that
$$
V@R_\lambda (X^{1,m})
=
\left\{
\begin{array}{ll}
V@R_\lambda(X), & \lambda < \alpha_1,\\
V@R_{\lambda + d - \alpha_1} (X), & \alpha_1 \leq \lambda < 1 - d + \alpha_1, \\
-m, & 1-d+\alpha_1 \leq \lambda. 
\end{array}
\right.
$$
By Theorem 6 in \ci{Dhaene2012} we obtain 
$$ \rho^1 (X^{1,m}) = \int _{[0,1]}   V@R_\lambda (X^{1,m})  g^1 (d \lambda) = c - m \cdot (1 - g^1 (1-d+\alpha_1))  $$
where the constant $c \geq 0$ is given by
$$c = \int _{[0,1]}  V@R_\lambda (X) \cdot \1 _{[0, \alpha_1)} (\lambda)  + V@R _ {\lambda + d - \alpha_1} (X) \cdot \1_{[\alpha_1, 1 - d+ \alpha_1)} (\lambda)   g^1( d \lambda)  < \infty .$$
By assumption, $1- g^1(1-d+\alpha_1) >0$ , thus $\rho^1(X^{1,m}) \to - \infty$ as $m \to \infty$.
\end{proof}

\subsection{Proof of Theorem \ref{thm:main3}}\label{proof:main3}

\begin{proof}
Using the notation of Theorem \ref{thm:main1}, we define for $i=1,2, \dots,n$ the random variables  $Z^i  = Y \cdot \1_{\{  \sum_{l=1}^{i-1} \alpha_l     \leq U < \sum_{l=1}^i \alpha_l         \}  }$. Then $Z^i \equiv 0$ on the set $\{  V@R_{\alpha_i} (Z^i) \geq Z^i  \}$ and $V@R_{\alpha_i} (Z^i) = 0$.  Thus,
\begin{eqnarray*}
\rho^i(X^i) & = & \rho^i \left(   Z^i + \frac {\essinf X}{n}  \right) \;  = \; \rho^i(Z^i)  \; + \;  \frac {\essinf X}{n} \\
& = & \rho^i\left(    Z^i \cdot \1 _{\{  V@R_{\alpha_i} (Z^i ) \geq Z ^i  \}}  + V@R_{\alpha_i} (Z^i)  \cdot \1 _{\{  V@R_{\alpha_i} (Z^i ) < Z ^i    \}}    \right) \; + \;  \frac {\essinf X}{n} \\
&= & \rho^i (0) \; + \; \frac {\essinf X}{n}
\end{eqnarray*}
This implies that $\sum_{i=1}^n \rho^i(X^i) = \sum_{i=1}^n \rho^i(0) + \essinf X  $.
\end{proof}

\begin{rem}
Instead of proving Theorem \ref{thm:main3} directly, we could also observe that $\rho^i \leq \rho^i (0) + V@R_{\alpha_i}$ and that the inf-convolution is montone increasing with respect to risk measures, and finally apply the corresponding results for $V@R_{\alpha_i}$, $i=1,2, \dots, n$, as given in \ci{embrechtsliuwang2016}. 
\end{rem}

\subsection{Proof of Theorem \ref{thm:smallsimple}}\label{sec:smallsimple}

\begin{proof}
Using the notation of Theorem  \ref{thm:main1}, we define for $i=1, 2, \dots, n-1$ and $m>0$ the random variables 
$$X^{i,m} \; = \;  (Y+m) \cdot \1 _{\{\sum_{l=1}^{i-1} \alpha_l  \leq U <  \sum_{l=1}^{i} \alpha_l\}}   + \frac{\essinf X}{n} $$
and 
\begin{eqnarray*}
X^{n,m } & =  & (Y+m) \cdot \1 _{\{ \sum_{l=1}^{n-1} \alpha_l \leq U <  d\}}   +  Y\cdot \1 _ {\{  d \leq U  \}} + \frac{\essinf X}{n}\\
X^{n+1,m} & = & -m \cdot \1_{\{ d > U  \}}.
\end{eqnarray*}
Since $\rho^1, \rho^2, \dots, \rho^n$ are V@R-type risk measures with parameters $\alpha_1, \alpha_2, \dots, \alpha_n$, we compute
\begin{eqnarray*}
\rho^i(X^{i,m})  & = & \rho^i(0) + \frac{\essinf X}{n}, \quad i=1,2, \dots, n-1,\\
\rho^n(X^{n,m})& = & \rho^n\left(V@R_d (Y) \cdot \1 _{\{ \sum_{l=1}^{n-1} \alpha_l \leq U <  d  \}} + Y \cdot \1 _ {\{   d \leq U  \}} + \frac{\essinf X}{n} \right)\\
&\leq & \rho ^n (0) + V@R_d(Y) + \frac{\essinf X}{n} 
\end{eqnarray*}
Since $\rho^{n+1}$ is strongly surplus sensitive at level $d$, we obtain that
$$ \rho^{n+1} (X^{n+1,m})  \leq  \rho^{n+1} (U - m \cdot \1_{\{V@R_{1-d} (U) \geq U   \}}) \stackrel{m \to \infty }{\longrightarrow} - \infty . $$
Thus,
$$\sum _{i=1}^{n+1} \rho^i(X^{i,m})  \;  \leq \;  \sum_{i=1}^n \rho^i(0) + V@R_d (Y) + \essinf X + \rho^{n+1} (X^{n+1,m})    \; \stackrel{m \to \infty }{\longrightarrow}   \;     - \infty . $$
\end{proof}

\vfill

\pagebreak

\bibliography{bibtex}
\bibliographystyle{jmr}

\end{document}